\documentclass[12pt,leqno,fleqn]{amsart}
\usepackage{latexsym}
\usepackage{amssymb}
\usepackage{amsxtra}
\usepackage{amsmath}
\usepackage{amsfonts}
\usepackage{CJK}
\newcommand{\cleqn}{\setcounter{equation}{0}}
\newcommand{\clth}{\setcounter{theorem}{0}}
\newcommand {\sectionnew}[1]{\section{#1}\cleqn\clth}
\newtheorem{theorem}{Theorem}[section]
\newtheorem{lemma}[theorem]{Lemma}
\newtheorem{proposition}[theorem]{Proposition}

\def\({\left(}
\def\){\right)}
\def\[{\begin{eqnarray}}
\def\]{\end{eqnarray}}
\def\d{\partial}

\def\d{\partial}

\def\La{\Lambda}
\def\la{\lambda}

\def\De{\Delta}

\newcommand{\Z}{\mathbb{Z}}

\makeatother       

\makeatletter      
\@addtoreset{equation}{section}
\makeatother       

\title{Ghost symmetry of the discrete KP hierarchy }
 \author{Chuanzhong Li\dag$^\star$ \ \ Jipeng Cheng\dag\dag \ \ Kelei Tian\ddag \ \ Maohua Li\dag \ \ Jingsong He\dag $^*$}

\dedicatory {  \dag Department of
Mathematics,  NBU, Ningbo, 315211, Zhejiang, P.\ R.\ China\\
  \dag\dag Department of
Mathematics,  CUMT, Xuzhou, 221116, Jiangsu, P.\ R.\ China\\
\ddag School of  Mathematical Sciences,  USTC, Hefei, 230026, Anhui, P.\ R.\ China\\
 $^\star$email:lichuanzhong@nbu.edu.cn\\
$^*$email:hejingsong@nbu.edu.cn
}

\thanks{$^*$ Corresponding author}

\date{}
\begin{document}

\begin{abstract}
 In this paper, with the help of the $S$ function and ghost symmetry for the discrete KP hierarchy which is a semi-discrete version  of the KP hierarchy, the ghost flow on its eigenfunction(adjoint eigenfunction) and the spectral representation of its Baker-Akhiezer function and adjoint Baker-Akhiezer function are derived. From these observations above, some important distinctions between the discrete KP hierarchy and KP hierarchy are shown. Also we give the ghost flow on the tau function and another kind of proof of the ASvM formula of the discrete KP hierarchy.

\end{abstract}


\maketitle
\allowdisplaybreaks
 \setcounter{section}{0}

\section{Introduction}

\setcounter{section}{1}
The discrete KP(dKP) hierarchy is an
interesting object in the research of  integrable systems \cite{Kupershimidt,Chaos,Iliev}. The discrete KP-hierarchy can be viewed
 as the classical KP-hierarchy \cite{DKJM,dickeybook} with the continuous derivative
$\frac{\d}{\d x}$ replaced formally by the discrete derivative $\Delta$ whose action on function $f(n)$ as
\[\triangle f(n)=f(n+1)-f(n).
\]

The Hamiltonian structures and
 tau function for the discrete KP hierarchy
was introduced in (\cite{Kupershimidt}-\cite{Iliev}). In \cite{LiuS}, the determinant representation of the gauge transformation for the discrete KP
hierarchy was introduced. The
Sato Backlund transformation, additional symmtries and ASvM formula for the discrete KP hierarchy was considered in \cite{LiuS2}.
The fermionic approach to Darboux transformations was considered in \cite{Willox} for the 1-component KP hierarchy and showed that any solution of the
associated (adjoint) linear problems can always be expressed as a superposition of
wave functions.
It is shown how the Darboux and binary Darboux transformations for a nonautonomous
discrete KP equation can be obtained from fermion analysis in \cite{WilloxJMP}.

 More recently, the extended discrete KP hierarchy and the algebraic structures
of the non-isospectral flows of the discrete KP hierarchy were investigated in \cite{YaoY}. The extended flow in \cite{YaoY} which is quite related to ghost symmetry inspired us to use the so-called ghost symmetry to derive some new results including the spectral  representation of the Baker-Akhiezer function and adjoint Baker-Akhiezer function in this paper.
After discretization, many important integrable properties are inherited by the discrete KP hierarchy from the KP hierarchy(\cite{SunXL}, \cite{adlerCMP},\cite{Dickey}).
For example: these include the existence of tau function, hamiltonian structure of the discrete KP hierarchy  and the close relationship between tau function of the discrete KP hierarchy and the one of  the KP hierarchy \cite{Iliev}. It is also interesting to further explore the new facts to show  the difference between  the discrete KP hierarchy and the KP hierarchy from the point of view of  symmetries.

As we all know, the symmetry is always an important research object in integrable systems. Many important characters of integrable systems have a close relation with symmetries, e.g. conserve laws, hamiltonian structures and so on. As one kind of additional symmetry, the ghost symmetry was discovered by W. Oevel \cite{WOevel}. After that, it attracts a lot of research (\cite{WOevelRMP} -\cite{Jipeng}). H. Aratyn  used the method
squared eigenfunction potentials(SEP or $S$ function later) to construct ghost symmetry of KP hierarchy and connect this kind of symmetry with constrained KP hierarchy(\cite{Aratyn97pla}-\cite{Aratyn}). W. Oevel and S. Carillo used $S$ function to represent $2+1$-dimensional hierarchies of the KP equation,
the modified KP equation and the Dym equation\cite{WOevel2,WOevel3}. J. P. Cheng etc. gave a good construction of the ghost symmetry of the BKP hierarchy\cite{Jipeng}.
In this paper, in order to study the action of ghost flows on the wave functions and tau function, it is necessary to introduce the $S$ function and spectral representation of the discrete KP hierarchy.  The spectral representation and  ghost symmetry   will tell us an important difference of the discrete KP hierarchy from the KP hierarchy.

The paper is organized as follows.  In Section 2, after recalling some basic facts of the discrete KP hierarchy (\cite{Kupershimidt}-\cite{Iliev}), the squared eigenfunction potentials ($S$ function) is introduced.   In Section 3,   the  ghost symmetry of the discrete KP hierarchy will be given with the help of the $S$ function. In Section 4, the spectral representation of eigenfunctions for the discrete KP hierarchy help us in deriving the ghost flow of eigenfunctions from the flow on Baker-Akhiezer wave functions, and meanwhile we give  some  nice properties of these functions.  Using these properties we give the ghost flow on the tau function and  a different proof of ASvM formula from\cite{LiuS2} in Section 5.
Section 6 is devoted to conclusions and
discussions.

\section{The discrete KP hierarchy and $S$ function }
\setcounter{section}{2}
To save the space in this section, we would like to follow reference \cite{Iliev,LiuS2} to recall some basic known  facts about the discrete KP hierarchy.
Firstly for an arbitrary difference function $ g(n)=g(n,t_1,t_2,\cdots,t_j,\cdots);  n\in\mathbb{Z}, t_i\in\mathbb{R},$
the shift operator  acting on this discrete function $g(n)$ is defined by
\[\Lambda g(n)=g(n+1).
\]
A difference operator $\triangle$ which acts on the function $g(n)$ is defined as following
\[\triangle g(n)=(\Lambda -1)g(n)=g(n+1)-g(n).
\]

The product rule of the $j$-th power  of operator $\triangle$ can be defined by operators' multiplication $``\circ"$,

\begin{equation}\triangle^k\circ
g=\sum^{\infty}_{i=0}\binom{k}{i}(\triangle^i
g)(n+k-i)\triangle^{k-i},
\end{equation}
where $\binom{k}{i}$ is the ordinary combinatorics number choosing $i$ from $k$.
For a formal pseudo
difference operators $G=\sum_{j=-\infty}^m g_j(n)\triangle^j, g_j(n)\in
R, n\in\mathbb{Z},$
and denote $G_+:=\sum_{j=0}^m g_j(n)\circ\triangle^j$ as the positive
projection of $G$ and by $G_-:=\sum_{j=-\infty}^{-1}
g_j(n)\circ \triangle^j$, the negative projection of $G$. Also the action on a function $g(n)$ by the
adjoint difference operator $\triangle^*$  is defined by,
\[\triangle^* g(n)=(\Lambda^{-1}-1)g(n)=g(n-1)-g(n),
\]
where $\Lambda^{-1} g(n)=g(n-1)$, and the corresponding product rule
operation is
\[\triangle^{*k}\circ
g=\sum^{\infty}_{i=0}\binom{k}{i}(\triangle^{*i}g)(n+i-k)\triangle^{*k-i}.
\]
Also  the formal adjoint $R^*$ to $R$ is defined
 as $R^*=\sum_{j=-\infty}^m
\triangle^{*j}\circ g_j(n)$. The adjoint $*$ operation satisfies a anti-involution rule as $(M\circ
N)^*=N^*\circ M^*$ for two arbitrary operators $M,N$ and $f(n)^*=f(n)$ for an arbitrary
function $f(n)$.

         The discrete KP hierarchy \cite{Iliev,LiuS2} is a family of evolution equations depending on
infinitely many variables $t=(t_1,t_2,\cdots)$
\begin{equation}
\frac{\partial L}{\partial t_i}=[B_i, L],\ \ \ B_i:=(L^i)_+,
\end{equation}

where $L$ is a general first-order pseudo difference operator
\begin{equation} \label{laxoperatordkp}
L(n)=\triangle + \sum_{j=0}^{\infty} u_j(n)\triangle^{-j}.
\end{equation}
Similar to the KP hierarchy, $L$  can also be  dressed by operator $W$\cite{LiuS2}
\[W(n;t)=1+\sum^\infty_{j=1}w_j(n;t)\triangle^{-j},
\]
by
\begin{equation}
L=W\circ\triangle\circ W^{-1}.\label{88}
\end{equation}
         There are the Baker-Akhiezer wave function $\Phi_{BA}(n;t,z)$ and adjoint Baker-Akhiezer wave
function $\Psi_{BA}(n-1;t,z)$ \cite{Iliev,LiuS2} constructed by,
\[\label{PhiBA}
\Phi_{BA}(n;t,z)&=&W(n;t)(1+z)^n e^{\sum^\infty_{i=1}t_i z^i},
\]
and
\[\label{PhiBA2}
\Psi_{BA}(n;t,z)&=&(W^{-1}(n-1;t))^*(1+z)^{-n} e^{\sum^\infty_{i=1}-t_i
z^i},
\]
which satisfy
\[
L(n)\Phi_{BA}(n;t,z)&=&z\Phi_{BA}(n;t,z),\ \ \ L^*(n-1)\Psi_{BA}(n;t,z)=z\Psi_{BA}(n;t,z),
\]
and
\[
\d_{t_j}\Phi_{BA}(n;t,z)&=&B_j(n)\Phi_{BA}(n;t,z),\ \ \ \d_{t_j}\Psi_{BA}(n;t,z)=-B_j^*(n-1)\Psi_{BA}(n;t,z).
\]
Then the Baker-Akhiezer wave function $\Phi_{BA}(n;t,z)$ and adjoint Baker-Akhiezer wave
function $\Psi_{BA}(n-1;t,z)$ will have forms as
\[
\Phi_{BA}(n;t,z)  &=&(1+\frac{w_1(n;t)}{z}+\frac{w_2(n;t)}{z^2}+\cdots)(1+z)^n e^{\sum^\infty_{i=1}t_i
        z^i}\label{41},
\]
and
\[
\Psi_{BA}(n;t,z)
&=&(1+\frac{w_1^*(n;t)}{z}+\frac{w_2^*(n;t)}{z^2}+\cdots)(1+z)^{-n}
e^{\sum^\infty_{i=1}-t_i z^i}.
\]
 Also there exists a tau function $\tau_\triangle=\tau(n;t)$  for the discrete KP
hierarchy \cite{Iliev,LiuS2}, which satisfies
\begin{equation}\label{taudefinition1}
1+\frac{w_1(n;t)}{z}+\frac{w_2(n;t)}{z^2}+\cdots=\frac{\tau(n;t-[z^{-1}])}{\tau(n;t)},
\end{equation}
and
\begin{equation}\label{taudefinition2}
1+\frac{w_1^*(n;t)}{z}+\frac{w_2^*(n;t)}{z^2}+\cdots=\frac{\tau(n;t+[z^{-1}])}{\tau(n;t)},
\end{equation}
where $[z^{-1}]=(\frac{1}{z},\frac{1}{2z^2},\frac{1}{3z^3},\cdots)$.

Now we prove some useful properties for the operators which are used later. \\
{\sl {\bf  Lemma 2.1}}  For $f\in F$ and $\triangle$, $\Lambda$ as
above, the following identities hold true.
\begin{align}
&(1)\quad \triangle\circ\Lambda=\Lambda\circ\triangle,  \\
&(2)\quad \triangle^*=-\triangle\circ\Lambda^{-1},  \\
&(3)\quad
(\triangle^{-1})^*=(\triangle^*)^{-1}=-\Lambda\circ\triangle^{-1}, \\
&(4)\quad \triangle^{-1}\circ
f\circ\triangle^{-1}=(\triangle^{-1}f)\circ\triangle^{-1}-\triangle^{-1}\circ
\Lambda(\triangle^{-1} f).\label{85}
\end{align}
\begin{proof}
The proof is standard and direct. We omit it here.
\end{proof}

To define S function(called Squared Eigenfunction Potential for KP hierarchy), we need the following proposition firstly similar as \cite{WOevel} which is about the constrained KP hierarchy. For an operator $A=\sum_{j\in \Z}A_j\Delta^j$, we define its residue $res\,  A=A_{-1}.$
\begin{proposition}\label{formula}
The following identities hold:
\[
res((\Delta A)(j))&=&res(\Delta \circ A(j)-A(j+1)\circ \Delta),\\
\label{leq0}
P_{< 0}(\De^{-1}A)
&=& \De^{-1} P_{< 0}(A)+ \De^{-1}P_{ 0}(A^*),\\
\label{resp0}
res(\Delta^{-1} A)&=&-(\La^{-1}P_0(A^*)),
\]
where $(\Delta A)$ denotes the action of $\Delta $ on operator $A$, $(\La^{-1}P_0(A^*))$ means a backward shift of function $P_0(A^*)$(zero order term of operator $A^*$ over $\De$) on discrete parameter.
\end{proposition}
\begin{proof}
Firstly the operator $A(j)$ is supposed to have the following form
\[
A(j)=A_n(j)\Delta^n+A_{n-1}(j)\Delta^{n-1}+\dots+A_{-1}(j)\Delta^{-1}+\dots,\]
where parameter $j$ denotes the discrete parameter.
Then it is easy to do the following calculation
\[
res((\Delta A)(j))=A_{-1}(j+1)-A_{-1}(j)=res(\Delta \circ A(j)-A(j+1)\circ \Delta).\]

If we rewrite operator $A$ into $\sum_{i\in \Z}\De^i\tilde A_i$, then \eqref{leq0} can be got as following

\[P_{< 0}(\De^{-1}A)=P_{< 0}(\De^{-1}\sum_{i\in \Z}\De^i\tilde A_i)
&=& \De^{-1} P_{< 0}(A)+ \De^{-1}P_{ 0}(A^*).
\]
Taking the residue of both sides of \eqref{leq0} will lead to \eqref{resp0}.
\end{proof}

Using above proposition, the following proposition can be easily got similarly as \cite{WOevel}.
\begin{proposition}
If $\alpha$ and $\beta$ are two local difference operators, then
\[res(\De^{-1}\alpha \beta \De^{-1} )&=&res(\De^{-1}P_0(\alpha^*) \beta \De^{-1} )+res(\De^{-1}\alpha P_0(\beta) \De^{-1} ).
\]
\end{proposition}
\begin{proof}
Using the Proposition \ref{formula}, the following calculation will lead to the proposition.
\begin{eqnarray*}
res(\De^{-1}\alpha \beta \De^{-1} )&=&res(\De^{-1}\alpha P_{\geq 1}(\beta) \De^{-1} )+res(\De^{-1}\alpha P_0(\beta) \De^{-1} )\\
&=&res(P_{< 0}(\De^{-1}\alpha) P_{\geq 1}(\beta) \De^{-1} )+res(\De^{-1}\alpha P_0(\beta) \De^{-1} )\\
&=&res(\De^{-1}P_{ 0}(\alpha^*) P_{\geq 1}(\beta) \De^{-1} )+res(\De^{-1}\alpha P_0(\beta) \De^{-1} )\\
&=&res(\De^{-1}P_{ 0}(\alpha^*) \beta \De^{-1} )+res(\De^{-1}\alpha P_0(\beta) \De^{-1} ).
\end{eqnarray*}
\end{proof}
This proposition will be used to prove the existence of the $S$ function. In the following, we sometimes denote $\phi_{t_n},\d_{t_n}\phi$ short for $\frac{\d \phi}{\d  t_n}$.

Till now, it is time to derive the existence of the $S$ function which is contained in the following proposition.
Before that, we need the following definition.
The functions $\phi,\psi$ which satisfy $\phi_{t_n}(j)=P_0(B_n(j)\phi(j))$ and $\psi_{t_n}(j)=-P_0(B^*_n(j)\psi(j))$ will be called the {\bf eigenfunction and adjoint eigenfunction} of the discrete KP hierarchy respectively.
 By the S function, the following proposition  can be got.
\begin{proposition}
For the eigenfunction $\phi$ and adjoint eigenfunction $\psi$ of the discrete KP hierarchy,  there exists a function $S(\psi,\phi)$, s.t.
\[\label{SDelta}S(\psi,\phi)_{\Delta }&=&\psi\phi,\ \ \\ \label{Stn}
S(\psi,\phi)_{t_n}&=&res(\De^{-1}\circ\psi\circ B_n\circ\phi \circ\De^{-1} ),
\]
where $S(\psi,\phi)_{\Delta }$ means the difference of function  $S(\psi,\phi)$ by the operator $\Delta$.
\end{proposition}

\begin{proof}
In the following proof, we temporarily omit symbol $``\circ"$ in the operator multiplication.
Eq.\eqref{SDelta} and eq.\eqref{Stn} can be rewritten as
\[\label{SDelta1}S(\psi(j),\phi(j))_{\Delta }&=&\psi(j) \phi(j),\\ \label{Stn1}
S(\psi(j),\phi(j))_{t_n}&=&res(\De^{-1}\psi(j) B_n(j)\phi(j) \De^{-1} ).
\]
The commutativity of eq.\eqref{SDelta} and eq.\eqref{Stn} can be proved as following
\begin{eqnarray*}
 S(\psi,\phi)_{t_n\De}&=&\De res(\De^{-1}\psi(j) B_n(j)\phi(j) \De^{-1} )\\
 &=&res(\psi B_n(j)\phi(j) \De^{-1}-\De^{-1}\psi(j+1) B_n(j+1)\phi(j+1))\\
&=&\psi P_0( B_n(j)\phi(j))+\phi^*(j) P_0( B^*_n(j)\psi^*(j))\\
&=&\psi(j) \phi(j)_{t_n}+\phi(j) \psi(j)_{t_n}\\
&=& S(\psi,\phi)_{\De t_n},
\end{eqnarray*}

\begin{eqnarray*}\notag
 &&S(\psi,\phi)_{t_nt_m}-S(\psi,\phi)_{t_mt_n}\\
 &=&[  res(\De^{-1}\psi_{t_m}(j) B_n(j)\phi(j) \De^{-1} )+ res(\De^{-1}\psi(j) (B_n)_{t_m}(j)\phi(j) \De^{-1} )\\
 &&+ res(\De^{-1}\psi(j) B_n(j)\phi_{t_m}(j) \De^{-1} )-  res(\De^{-1}\psi_{t_n}(j) B_m(j)\phi(j) \De^{-1} )\\
 &&- res(\De^{-1}\psi(j) (B_m)_{t_n}(j)\phi(j) \De^{-1} )- res(\De^{-1}\psi(j) B_m(j)\phi_{t_n}(j) \De^{-1} )]\\
 &=& [- res(\De^{-1} P_0( B^*_m(j)\psi(j)) B_n(j)\phi(j) \De^{-1} )+ res(\De^{-1}\psi(j) (B_n)_{t_m}(j)\phi(j) \De^{-1} )
\\
&&+ res(\De^{-1}\psi B_n(j)P_0( B_m(j)\phi(j)) \De^{-1} )+  res(\De^{-1}P_0( B^*_n(j)\psi(j)) B_m(j)\phi(j) \De^{-1} )
\\
 &&- res(\De^{-1}\psi(j) (B_m)_{t_n}(j)\phi(j) \De^{-1} )- res(\De^{-1}\psi(j) B_m(j)P_0( B_n(j)\phi(j)) \De^{-1} )]\\
&=&[ res(\De^{-1}\psi(j) [(B_n)_{t_m}-(B_m)_{t_n}+[B_n,B_m]](j)\phi(j) \De^{-1} )]\\
&=&0.
\end{eqnarray*}
This is the end of the proof.
\end{proof}

Because of the following formula
\[f_1\De^{-1}g_1\circ f_2\De^{-1}g_2=f_1S(g_1,f_2)\De^{-1}g_2-f_1\De^{-1}\circ \La(S(g_1,f_2))g_2,\]
we can get some properties of the difference operator in the following proposition.

\begin{proposition}
The following identities hold true
\[
S(S(\La(g_1),\La(f_2))g_2,f_3)=S(S(g_1,f_2)\La^{-1}(g_2),\La^{-1}(f_3)),\\
S(\La(S(g_1,f_2))g_2,f_3)+\La(S(g_1,f_2S(g_2,f_3)))=\La(S(g_1,f_2)) \La(S(g_2,f_3)),
\]
for arbitrary functions $g_1,g_2,f_2,f_3$.
\end{proposition}

\begin{proof}

Direct calculation will lead to the following two identities
\[\notag
&&(f_1\De^{-1}g_1\circ f_2\De^{-1}g_2)\circ f_3\De^{-1}g_3\\ \notag
&=&f_1S(g_1,f_2)S(g_2,f_3)\De^{-1}g_3-f_1S(g_1,f_2)\De^{-1}\circ\La(S(g_2,f_3))g_3\\ \notag
&&-f_1S(\La(S(g_1,f_2))g_2,f_3)\De^{-1}g_3+f_1\De^{-1}\circ S(\La(S(g_1,f_2))g_2,f_3)g_3,
\]
\[\notag
&&f_1\De^{-1}g_1\circ (f_2\De^{-1}g_2\circ f_3\De^{-1}g_3)\\ \notag
&=&f_1S(g_1,f_2S(g_2,f_3))\De^{-1}g_3-f_1\De^{-1}\circ\La(S(g_1,f_2S(g_2,f_3)))g_3\\ \notag
&&-f_1S(g_1,f_2)\De^{-1}\circ \La(S(g_2,f_3))g_3+f_1\De^{-1}\circ\La(S(g_1,f_2)) \La(S(g_2,f_3))g_3.
\]

Let \[(f_1\De^{-1}g_1\circ f_2\De^{-1}g_2)\circ f_3\De^{-1}g_3=f_1\De^{-1}g_1\circ (f_2\De^{-1}g_2\circ f_3\De^{-1}g_3)\]
 and compare the $\De^{-1}$ term, we can get

\[\label{three}
S(S(\La(g_1),\La(f_2))g_2,f_3)=S(S(g_1,f_2)\La^{-1}(g_2),\La^{-1}(f_3)).
\]

Comparing the $\De^{-2}$ terms of both sides of eq.\eqref{three}, we can get

\[\notag
S(\La(S(g_1,f_2))g_2,f_3)+\La(S(g_1,f_2S(g_2,f_3)))=\La(S(g_1,f_2)) \La(S(g_2,f_3)).
\]

\end{proof}

 Using this $S$ function, the ghost symmetry can be constructed in the next section.

\sectionnew{The ghost symmetry of the discrete KP hierarchy}
In this section,  the ghost flows on Lax operator  of discrete KP hierarchy will be introduced firstly. Then one can prove that they are symmetries of the discrete KP hierarchy. After this, we naturally  further consider the action of ghost flows on wave function(Baker-Akhiezer wave function). In the following part, we always omit the same discrete parameter in one equation.

Inspired  by the definition of the ghost symmetry (or squared eigenfunction symmetry) flow of the KP hierarchy\cite{YaoY,WOeveldarboux,Aratyn},  the discrete ghost flow of the discrete KP hierarchy can be defined as following
\[\label{ghostflow}
\d_Z L=[\phi \De^{-1}\psi, L],\]
where functions $\phi,\psi$  are the eigenfunction and adjoint eigenfunction  of  the discrete KP hierarchy. They correspond to the same $B_n$. Therefore we leave the discrete parameter out in the following part. In fact, this ghost flow can also be derived from extended discrete KP hierarchy \cite{YaoY}.
The following proposition will tell you the above flow is a symmetry of discrete KP hierarchy.
\begin{proposition}\label{symmetry}
The additional flows $\d_Z$ commute
with the discrete KP flows $\d_{ t_n}$, i.e.,
\begin{eqnarray}
[\d_Z, \d_{ t_n}]L=0,
\end{eqnarray}
\end{proposition}
\begin{proof}
The commutativity of ghost flows and the discrete KP flows are in fact equivalent to the following Zero-Curvature equation which includes detailed proof
\begin{eqnarray*}&&\d_ZB_n-\d_{ t_n}( \phi \De^{-1}\psi)+[B_n, \phi \De^{-1}\psi]\\
&=&[ \phi \De^{-1}\psi, L^n]_+- \phi_{t_n} \De^{-1}\psi- \phi \De^{-1}\psi_{t_n}+[B_n, \phi \De^{-1}\psi]\\
&=&[ \phi \De^{-1}\psi,B_n]_+-P_0(B_n\phi) \De^{-1}\psi+\phi \De^{-1}P_0(B^*_n\psi)+[B_n,\phi \De^{-1}\psi]\\
&=&( B_n \phi \De^{-1}\psi)_--(\phi \De^{-1}\psi B_n)_--P_0(B_n\phi) \De^{-1}\psi+\phi \De^{-1}P_0(B^*_n\psi)\\
&=&0.\end{eqnarray*}

\end{proof}

The ghost symmetry on the wave operator $W$ can be got as following
\[\d_Z W=\phi \De^{-1}\psi W.\]
According to  eq.\eqref{ghostflow}, the ghost flows acting on the
Baker-Akhiezer function $\Phi_{BA}(n;t,z)$ and the adjoint Baker-Akhiezer function $\Psi_{BA}(n;t,z)$ are in the following proposition.
\begin{proposition}\label{flowonBA}
The Baker-Akhiezer function $\Phi_{BA}(n;t,z)$ and adjoint Baker-Akhiezer function $\Psi^*_{BA}(n;t,z)$ satisfy the following equations
\[\label{dZPhi}\d_Z\Phi_{BA}(n;t,z)&=&\phi S(\psi,\Phi_{BA}(n;t,z)),\\
\d_Z\Psi_{BA}(n;t,z)&=&\psi S(\La(\phi),\Psi_{BA}(n;t,z)).
\]
\end{proposition}

\begin{proof}

The ghost symmetry on wave operator $W^{-1}$ and $W^{-1*}$ can be got as following

\[\d_Z W^{-1}=-W^{-1}\phi \De^{-1}\psi, \]
\[\d_Z W^{-1*}=\psi \La\De^{-1}\phi W^{-1*}. \]

Considering eq.\eqref{PhiBA} and eq.\eqref{PhiBA2} and taking derivative by $\d_Z$ will lead to the following calculation which will finish the proof.
\begin{align*}
\d_Z\Phi_{BA}(n;t,z)&=(\d_ZW(n;t))(1+z)^n exp(\sum^\infty_{i=1}t_i z^i)\notag\\
        &=\phi \De^{-1}\psi W(1+z)^n exp(\sum^\infty_{i=1}t_i z^i)\\
        &=\phi S(\psi,\Phi_{BA}(n;t,z)),
\end{align*}
\begin{align*}
\d_Z\Psi_{BA}(n;t,z)&=(\d_ZW^{-1*}(n-1;t))(1+z)^{-n} exp(\sum^\infty_{i=1}-t_i
z^i)\notag\\
&=\psi \La\De^{-1}\phi(W^{-1}(n-1;t))^*(1+z)^{-n} exp(\sum^\infty_{i=1}-t_i
z^i)\notag\\
&=\psi \De^{-1}\La(\phi)(W^{-1}(n;t))^*(1+z)^{-n-1} exp(\sum^\infty_{i=1}-t_i
z^i)\notag\\
&=\psi S(\La(\phi),\Psi_{BA}(n+1;t,z)).
\end{align*}
\end{proof}
From above, the $S$ function\cite{Aratyn} is not used directly in the definition of the ghost flow on the Baker-Akhiezer function $\Phi_{BA}(n;t,z)$ and adjoint Baker-Akhiezer function $\Psi_{BA}(n;t,z)$.

Till now, the action of the ghost flow on the Baker-Akhiezer function $\Phi_{BA}(n;t,z)$ and the adjoint Baker-Akhiezer function $\Psi_{BA}(n;t,z)$ is derived. Therefore it is natural to further consider the ghost flow on the tau function. Before that we need some preparation which is contained in the next section.

\sectionnew{Properties and Spectral representation of discrete KP hierarchy}
This section will be about the spectral representation of eigenfunction for  the discrete KP hierarchy. Before that, we firstly introduce some properties of the tau function and wave functions of the discrete KP hierarchy. These all give a good preparation to derive the ghost flow on the tau function and a new proof of the ASvM formula which will be given in the next section.

The Hirota bilinear equations of discrete KP hierarchy are as following \cite{Iliev}.
\begin{lemma}
 The Hirota bilinear identities of the discrete KP hierarchy are
\[\label{hirotabi}
res_z (\De^j\Phi_{BA}(n,t',z))\Psi_{BA}(n,t,z)=0,\ \ j\geq 0.
\]
\end{lemma}
Using this bilinear identity and relation between the tau function and the Baker-Akhiezer wave function and adjoint Baker-Akhiezer wave function, the following proposition can be got\cite{sato,LiuS2}.
\begin{proposition}
The tau functions of discrete KP hierarchies satisfy the following Fay identity\cite{sato,LiuS2}
\[
(s_0-s_1)(s_2-s_3)\tau(n,t+[s_0]+[s_1])\tau(n,t+[s_2]+[s_3])+c.p.=0,
\]
where the cyclic permutation is only over $s_1$, $s_2$ and $s_3$.
\end{proposition}
Set $s_0=0$ and divide by $s_1s_2s_3$, then shift time variables by $[s_2]+[s_3]$ and we get the following identity
\[
&&(s_2^{-1}-s_3^{-1})\tau(n,t+[s_1]-[s_2]-[s_3])\tau(n,t)\\ \notag
&&+(s_1^{-1}-s_2^{-1})\tau(n,t-[s_2])\tau(n,t+[s_1]-[s_3])\\ \notag
&&+(s_3^{-1}-s_1^{-1})\tau(n,t-[s_3])\tau(n,t+[s_1]-[s_2])=0.
\]

Because $\tau(n+1,t)=\tau(n,t-[-1])$\cite{LiuS2}, the following proposition can  be got.
\begin{proposition}\label{fay1}
The tau functions of the discrete KP hierarchy satisfy
\[
&&(1+s_3^{-1})\La(\frac{\tau(n,t+[s_1]-[s_3])}{\tau(n,t)})\\ \notag
&&=(s_3^{-1}-s_1^{-1})\frac{\tau(n,t-[s_3])\tau(n+1,t+[s_1])}{\tau(n,t)\tau(n+1,t)}
+(1+s_1^{-1})\frac{\tau(n,t+[s_1]-[s_3])\tau(n+1,t)}{\tau(n,t)\tau(n+1,t)}.
\]
\end{proposition}
The Proposition \ref{fay1} can be rewritten as the following new proposition \cite{LiuS2}.
\begin{proposition}\label{Fay}
The tau functions of discrete KP hierarchies satisfy the following difference Fay identity
\[
&&(1+s_3^{-1})\De(\frac{\tau(n,t+[s_1]-[s_3])}{\tau(n,t)})\\ \notag
&&=(s_3^{-1}-s_1^{-1})(\frac{\tau(n,t-[s_3])\tau(n+1,t+[s_1])}{\tau(n,t)\tau(n+1,t)}
-\frac{\tau(n,t+[s_1]-[s_3])}{\tau(n,t)}).
\]
\end{proposition}
Similar as \cite{Aratyn}, we can get the following property of the Baker-Akhiezer wave function and adjoint Baker-Akhiezer wave function.
\begin{proposition}\label{deltapro}
The following identity holds
\[\label{1}
\De(\Phi_{BA}(n,t+[\la^{-1}],\mu)\Psi_{BA}(n;t,\la))=\la\Psi_{BA}(n+1,t,\la)\Phi_{BA}(n,t,\mu),\]
\[\label{2}
\De(\Phi_{BA}(n,t,\mu)\Psi_{BA}(n;t-[\mu^{-1}],\la))=-\mu\Psi_{BA}(n+1,t,\la)\Phi_{BA}(n,t,\mu).\]
\end{proposition}
\begin{proof}

By Proposition \ref{Fay}, we can get the following  identity
\[\notag
&&\frac{(1+\mu)^{n+1}}{(1+\la)^{n+1}}e^{\xi(t,\mu)-\xi(t,\la)}\frac{1}{1-\frac{\mu}{\la}}\frac{\tau(n+1,t+[\la]-[\mu])}{\tau(n+1,t)})\\
\notag
&&
-\frac{(1+\mu)^n}{(1+\la)^{n}}e^{\xi(t,\mu)-\xi(t,\la)}\frac{1}{\frac{\la}{\mu}-1}\frac{\tau(n,t+[\la]-[\mu])}{\tau(n,t)}\\ \notag
&&=\frac{(1+\mu)^n}{(1+\la)^{n+1}}e^{\xi(t,\mu)-\xi(t,\la)}\frac{\tau(n,t-[\mu])\tau(n+1,t+[\la])}{\tau(n,t)\tau(n+1,t)}.
\]

Then it further leads to the following identity
\[\label{identiey}
&&\frac{1}{\la-\mu}\De(\frac{\la^n}{\mu^n}\frac{(1+\mu)^n}{(1+\la)^{n}}e^{\xi(t,\mu)-\xi(t,\la)}\frac{\tau(n,t+[\la]-[\mu])}{\tau(n,t)})\\ \notag
&=&\frac{\la^n}{\mu^{n+1}}\frac{(1+\mu)^n}{(1+\la)^{n+1}}e^{\xi(t,\mu)-\xi(t,\la)}\frac{\tau(n,t-[\mu])\tau(n+1,t+[\la])}{\tau(n,t)\tau(n+1,t)}.\]

The equation \eqref{identiey} above can have the following form if we change $\la$ and $\mu$ to  $\la^{-1}$ and $\mu^{-1}$
\[\label{identiey2}
&&\frac{1}{\la-\mu}\De(\frac{(1+\mu)^n}{(1+\la)^{n}}e^{\xi(t,\mu)-\xi(t,\la)}\frac{\tau(n,t+[\la^{-1}]-[\mu^{-1}])}{\tau(n,t)})\\ \notag
&=&\frac{(1+\mu)^n}{(1+\la)^{n+1}}e^{\xi(t,\mu)-\xi(t,\la)}\frac{\tau(n,t-[\mu^{-1}])\tau(n+1,t+[\la^{-1}])}{\tau(n,t)\tau(n+1,t)}.\]

Denoting
\[\label{Xdef}X(n,\la,\mu):=\frac{(1+\mu)^n}{(1+\la)^{n}}e^{\xi(t,\mu)-\xi(t,\la)}e^{\sum_{l=1}^{\infty}\frac{1}{l}(\la^{-l}-\mu^{-l})\frac{\d}{\d t_l}},\]
 we can get the following result
\[\frac{X(n,\la,\mu)\tau(n,t)}{\tau(n,t)}&=&\frac{(1+\mu)^n}{(1+\la)^{n}}e^{\xi(t,\mu)-\xi(t,\la)}\frac{\tau(n,t+[\la^{-1}]-[\mu^{-1}])}{\tau(n,t)}\\
\notag &=&(1-\frac{\mu}{\la})\Phi_{BA}(n,t+[\la^{-1}],\mu)\Psi_{BA}(n;t,\la).\]
Similarly, we can also get
\[\frac{X(n,\la,\mu)\tau(n,t)}{\tau(n,t)}
&=&(1-\frac{\la}{\mu})\Phi_{BA}(n,t,\mu)\Psi_{BA}(n;t-[\mu^{-1}],\la),\]
by considering another different convergence field.

 Therefore we further get

\[\label{formula2}
\frac{1}{\la-\mu}\De(\frac{X(n,\la,\mu)\tau(n,t)}{\tau(n,t)})=\Psi_{BA}(n+1,t,\la)\Phi_{BA}(n,t,\mu),\]
which further leads to eq.\eqref{1} in the proposition.
In fact the equation \eqref{formula2} can also be obtained from formula (16) in \cite{WilloxJMP}, which itself is obtained from the general formula (4) by the discretization procedure explained in \cite{WilloxJMP}.

In the same way, eq.\eqref{2} can be easily proved.

\end{proof}

Besides the above proposition, the following lemma can also be got easily whose continuous version can be found in \cite{Aratyn}.
\begin{lemma}\label{Delataz}
The following identity holds
\[\notag
\frac{1}{\mu}\hat\De_z(\Phi_{BA}(n,t+[\mu^{-1}],\la)\Psi_{BA}(n;t,\mu))=\frac1z\Phi_{BA}(n,t,\la)\Psi_{BA}(n;t-[z^{-1}],\mu)
\]
where
\[\hat\De_zf(t)=f(t-[z^{-1}])-f(t).\]
\end{lemma}
\begin{proof}
According eq.\eqref{taudefinition1} and eq.\eqref{taudefinition2}, we can get the following identity
\[\notag
\Phi_{BA}(n,t+[\mu^{-1}],\la)\Psi_{BA}(n;t,\mu)=\frac{(1+\la)^n}{(1+\mu)^n}\frac{1}{1-\frac{\la}{\mu}}
e^{\xi(t,\la)-\xi(t,\mu)}\frac{\tau(n,t+[\mu^{-1}]-[\la^{-1}])}{\tau(n,t)}.
\]
 Then we can derive the following calculation which can finish the proof,

\begin{eqnarray*}
&&\frac{1}{\mu}\hat\De_z(\Phi_{BA}(n,t+[\mu^{-1}],\la)\Psi_{BA}(n;t,\mu))\\
&=&\frac{(1+\la)^n}{(1+\mu)^n}\frac{1}{\mu-\la}\frac{1-\frac{\la}{z}}
{1-\frac{\mu}{z}}
e^{\xi(t,\la)-\xi(t,\mu)}\frac{\tau(n,t+[\mu^{-1}]-[\la^{-1}]-[z^{-1}])}{\tau(n,t-[z^{-1}])}\\
&&-\frac{(1+\la)^n}{(1+\mu)^n}\frac{1}{\mu-\la}
e^{\xi(t,\la)-\xi(t,\mu)}\frac{\tau(n,t+[\mu^{-1}]-[\la^{-1}])}{\tau(n,t)}\\
&=&\frac{(1+\la)^n}{(1+\mu)^n}\frac{1}{(z-\mu)}
e^{\xi(t,\la)-\xi(t,\mu)}\frac{\tau(n,t-[\la^{-1}])\tau(n,t+[\mu^{-1}]-[z^{-1}])}{\tau(n,t)\tau(n,t-[z^{-1}])}\\
&=&\frac1z\Phi_{BA}(n,t,\la)\Psi_{BA}(n;t-[z^{-1}],\mu).
\end{eqnarray*}

\end{proof}

We need to remark that a few of the formulas in Proposition \ref{deltapro} and Lemma \ref{Delataz}  are in fact special cases of the general formulas that are obtained in \cite{Willox}, which deals with the fermonic representation of the KP hierarchy.

In the following part, it is time to introduce the spectral representation of eigenfunctions of the discrete KP hierarchy which will be shown in the following proposition.
From now on, we  denote above residue $res_z$ in eq.\eqref{hirotabi} as integral $\int$ in the following part.
\begin{proposition}\label{spectral}
The eigenfunction $\phi(n;t)$ and adjoint eigenfunction  $\psi(n;t)$ have the following spectral representation using the Baker-Akhiezer function $\Phi_{BA}(n;t,z)$ and adjoint Baker-Akhiezer function $\Psi_{BA}(n;t,z)$
\[\label{spectral}&&\phi(n,t)=-\int dz \bar \phi(n,z)(\Phi_{BA}(n;t,z)),\ \ \\
\label{spectral2}&&\psi(n;t)=\int dz \bar \psi(n,z)\Psi_{BA}(n+1;t,z),\ \
\]
where two density functions
$\bar \phi(n,z)$ and $\bar \psi(n,z)$ have the following spectral representation
\[\notag \bar \phi(n,z)=S(\phi(n,t'),\Psi_{BA}(n+1;t',z)),\ \ \ \bar \psi(n,z)=S(\psi(n,t'),(\Phi_{BA}(n+1;t',z))).\]

\end{proposition}

\begin{proof}
The eq.\eqref{hirotabi} can be rewritten as
\[\label{bilinear}\int dz (\Phi_{BA}(m;t,z))\Psi_{BA}(n;t',z))=0, \ \ m,n\in \Z_+.\]
In this proof, we choose $n=m+1.$
Deriving the right side of the eq.\eqref{spectral} and  using eq.\eqref{hirotabi}, we can find
\[&&\d_{t'_{m}}(\int dz (\Phi_{BA}(n;t,z))S(\phi(n,t'),\Psi_{BA}(n+1;t',z)))=0.\]
That means the right side of eq.\eqref{spectral} do not depend on $t'$.
Set $t'=t$, then the right side of eq.\eqref{spectral} becomes
\[\notag -\int dz (\Phi_{BA}(n;t,z))S(\phi(n,t),\Psi_{BA}(n+1;t,z))=\int dz \phi(n,t)(z^{-1}+o(z^{-2}))=\phi(n,t).\]
So the eq.\eqref{spectral} is proved.
To prove eq.\eqref{spectral2}, we need to choose $m=n$ in eq.\eqref{bilinear}. The process is quite similar, therefore we omit it here.
\end{proof}

From above spectral representation, one can see the difference between discrete KP hierarchy and KP hierarchy which shows the discrete effect.

The spectral representation will help us to get the ghost flow of eigenfunction $\phi(n;t)$ and adjoint eigenfunction  $\psi(n;t)$.
Considering the Proposition \ref{deltapro} and above spectral representation, this will lead to the following proposition.
\begin{proposition}
The following identities hold
\[\label{lamu}
\ \ \ \ \ \ \ \ \ \ \ S(\Psi_{BA}(n+1,t,\la),\Phi_{BA}(n,t,\mu))&=&\la^{-1}\Phi_{BA}(n,t+[\la^{-1}],\mu)\Psi_{BA}(n;t,\la),\\
\ \  \ \ \ \ \ \ \ \ \ \label{mula}S(\Psi_{BA}(n+1,t,\la),\Phi_{BA}(n,t,\mu))&=&-\mu^{-1}\Phi_{BA}(n,t,\mu)\Psi_{BA}(n;t-[\mu^{-1}],\la),\\
\label{Hphi}S(\Psi_{BA}(n+1,t,\la),\phi(n,t))&=&\la^{-1}\phi(n,t+[\la^{-1}])\Psi_{BA}(n;t,\la),\\
\label{Hpsi}S(\psi(n,t),\Phi_{BA}(n,t,\la))
&=&-\la^{-1}\Phi_{BA}(n,t,\la)\psi(n-1;t-[\la^{-1}]),\]
\[S(\psi(n,t),\phi(n,t))
&=&\int\int d\la d\mu \bar \phi(\mu)\bar \psi(\la)S(\Psi_{BA}(n+1,t,\la),\Phi_{BA}(n,t,\mu)).
\]
\end{proposition}
We need note that eq.\eqref{lamu} and eq.\eqref{mula} are defined in two different field.

After the proposition above about spectral representation of the discrete KP hierarchy, it is easy to lead to the ghost flow of the eigenfunction $\phi(n;t)$ and adjoint eigenfunction  $\psi(n;t)$ as the following proposition.

\begin{proposition}\label{flowoneigenfunction}
The eigenfunction $\phi(n;t)$ and adjoint eigenfunction  $\psi(n;t)$ satisfy the following equation
\[\d_Z\phi(n;t)&=&\phi S(\psi,\phi(n;t)),\\
\d_Z\psi(n;t)&=&\psi S(\La(\phi),\psi(n;t)).
\]
\end{proposition}

\begin{proof}
A direct calculation will lead to the following identities
\begin{align}\notag
\d_Z\phi(n;t)&=\int dz \bar \phi(z)\d_Z\Phi_{BA}(n;t,z)\\ \notag
        &=\phi \De^{-1}\psi\int dz \bar \phi(z)\Phi_{BA}(n;t,z)\\ \notag
        &=\phi S(\psi,\phi(n;t)),
\end{align}
\begin{align}\notag
\d_Z\psi(n;t)&=\int dz \bar \psi(z)\d_Z\Psi_{BA}(n+1;t,z)\\ \notag
&=\psi \La\De^{-1}\phi\int dz \bar \psi(z)\Psi_{BA}(n+1;t,z)\\ \notag
&=\psi S(\La(\phi),\psi(n;t)).
\end{align}
Then we finish the proof.
\end{proof}

\begin{proposition}
The following identity holds
\[\label{deltazH}
\hat\De_z(S(\phi(n,t),\psi(n;t))=-\frac{1}z\phi(n,t)\psi(n-1;t-[z^{-1}])
\]

where
\[\hat\De_zf(t)=f(t-[z^{-1}])-f(t).\]
\end{proposition}
\begin{proof}

Considering the Lemma \ref{Delataz} and eq.\eqref{lamu}, eq.\eqref{mula} can lead to equation as

\[\label{BAdeltaz}
\hat\De_z(S(\Phi_{BA}(n,t,\la),\Psi_{BA}(n+1;t,\mu)))=-\frac{1}z\Phi_{BA}(n,t,\la)\Psi_{BA}(n;t-[z^{-1}],\mu).
\]
Then using the spectral representation of the eigenfunction and adjoint eigenfunction, we can derive the proposition directly.

\end{proof}

\sectionnew{Ghost flows on the tau function and ASvM formula}
 After the good preparation in the last section, this section will be devoted to derive the ghost flow on the tau function and a new proof of the ASvM formula.

All ghost flows above is about \[\d_Z :=\phi \De^{-1}\psi\]
where the choices of  a pair $(\phi,\psi)$ are in the set $(E,E')$ of the eigenfunctions and adjoint eigenfunctions of discrete KP hierarchy. Also $\phi,\psi$ should correspond to the same $B_n(j)$ in the definition of the eigenfunctions and adjoint eigenfunctions of discrete KP hierarchy.  That means the ghost symmetry of discrete KP hierarchy can be generalized to
\[\d_Z :=\sum_{\phi\in E,\psi \in E'}\phi \De^{-1}\psi.\]
Now we will think about one specific ghost flow denoted as
\[\d_{Zn}:=\phi(n) \De^{-1}\psi(n).\]

 After above preparation, it is time to derive the ghost flow acting on the tau function of the discrete KP hierarchy which is contained in the following proposition.
\begin{proposition}\label{flowontau}
The ghost flow of discrete KP hierarchy on its tau function is as  the following
 \[\label{ghosttau}
 \d_{Zn}\tau(n;t)&=&-S(\phi(n,t),\psi(n;t))\tau(n;t),
\]
\end{proposition}
\begin{proof}
The eq.\eqref{dZPhi}, eq.\eqref{Hpsi} and eq.\eqref{deltazH} will lead to the following calculation
\[\notag\d_{Zn}\Phi_{BA}(n;t,z)&=&\phi(n) S(\psi(n),\Phi_{BA}(n;t,z))\\
\notag&=&z^{-1}\phi(n)\psi(n;t-[z^{-1}])\Phi_{BA}(n,t,z)\\
\notag&=&-\hat\De_z(S(\phi(n,t),\psi(n;t)))\Phi_{BA}(n,t,z).
\]
 Therefore we get
 \[\d_{Zn}\log\tau(n;t)&=&-S(\phi(n,t),\psi(n;t)),
\]
which further lead to
eq.\eqref{ghosttau}.
\end{proof}

One can easily check that the following identity holds using methods in \cite{dickeybook,dickeyadditional} which is about KP hierarchy
\[\label{XPHI}X(n,\la,\mu)\Phi_{BA}(n,t,z)&=&\Phi_{BA}(n,t,z)\hat\Delta_z\frac{X(n,\la,\mu)\tau(n,t)}{\tau(n,t)},\]
where $X(n,\la,\mu)$ is defined as \eqref{Xdef} whose action on $\Phi_{BA}(n,t,z)$ is an infinitesimal action on the tau function.

 After defining $Y(n,\la,\mu):=(\la-\mu)\Phi_{BA}(n,t,\mu)\Delta^{-1}\Psi_{BA}(n+1,t,\la)$ and using several above propositions and  lemmas, we will give a proof of  the ASvM formula as the following proposition in  a different way from \cite{LiuS2}.
\begin{proposition}
The action by the vertex operator $X(n,\la,\mu)$ as an infinitesimal transformation on the Baker-Akhiezer function $\Phi_{BA}(n;t,z)$ can be equivalently expressed by the action of infinitesimal operator $Y(n,\la,\mu)$, i.e.
  \[X(n,\la,\mu)\Phi_{BA}(n,t,z)=Y(n,\la,\mu)\Phi_{BA}(n,t,z).\]
\end{proposition}
\begin{proof}
Using eq.\eqref{formula2} and eq.\eqref{XPHI}, the following
equation can be got
\[\notag X(n,\la,\mu)\Phi_{BA}(n,t,z)&=&\Phi_{BA}(n,t,z)(\la-\mu)\hat\Delta_zS(\Phi_{BA}(n,t,\mu),\Psi_{BA}(n+1,t,\la)),\]
which further leads to the following identity by considering eq.\eqref{BAdeltaz}
\[\notag X(n,\la,\mu)\Phi_{BA}(n,t,z)&=&(\la-\mu)\Phi_{BA}(n,t,z)z^{-1}\Phi_{BA}(n,t,\mu)\Psi_{BA}(n,t-[z^{-1}],\la).\]
Then eq.\eqref{mula} can help us  deriving  the following identity
\[\notag X(n,\la,\mu)\Phi_{BA}(n,t,z)
&=&(\la-\mu)\Phi_{BA}(n,t,\mu)\Delta^{-1}\Psi_{BA}(n+1,t,\la)\Phi_{BA}(n,t,z),
\]
which is exactly what we need to prove.
This is the end of the simple proof for the ASvM formula.
\end{proof}

\sectionnew{Conclusions and Discussions}

In this paper, using ghost symmetry of the discrete KP hierarchy acting on Lax operator $L$, ghost flows on the Baker-Akhiezer function $\Phi_{BA}(n;t,z)$, adjoint Baker-Akhiezer function $\Psi_{BA}(n;t,z)$ are given in Proposition \ref{flowonBA}.  By spectral representation in terms of the Baker-Akhiezer function $\Phi_{BA}(n;t,z)$ and adjoint Baker-Akhiezer function $\Psi_{BA}(n;t,z)$ in Proposition \ref{spectral}, we derive ghost flows of the eigenfunction $\phi(n;t)$ and adjoint eigenfunction  $\psi(n;t)$ in Proposition \ref{flowoneigenfunction}. Meanwhile some nice properties of the Baker-Akhiezer function $\Phi_{BA}(n;t,z)$ and adjoint Baker-Akhiezer function $\Psi_{BA}(n;t,z)$ are also got with the help of Fay-identities. By these properties, ghost flows on the tau function are derived nicely  in Proposition \ref{flowontau}. Also we give a new proof of the ASvM formula  with the help of  the $S$ function.

Our next step is to connect these ghost flows with the discrete constrained KP hierarchy. Another interesting problem is to consider ghost flows of sub-hierarchies of the discrete KP hierarchy including the discrete BKP hierarchy and discrete CKP hierarchy. The difficulty is to identify the discrete algebraic structure hidden in the discrete KP hierarchy such that we can find a suitable reduction  to define sub-hierarchies of the discrete KP hierarchy.

{\bf Acknowledgments} {\noindent \small
We are grateful to Prof. Folkert Mueller-Hoissen in Max-Planck-Institute for Dynamics and Self-Organization (G\"ottingen in Germany)
for valuable discussions and suggestions.
Chuanzhong Li is supported by the National Natural Science Foundation of China under Grant No. 11201251,
 the Zhejiang Provincial Natural Science Foundation under Grant No. LY15A010004, LY12A01007, the Natural Science Foundation of Ningbo under Grant No. 2015A610157, 2013A610105, 2014A610029. Maohua Li is supported by  the Zhejiang Provincial Natural Science Foundation under Grant No. LY15A010005. Jingsong He is supported by the National Natural Science Foundation of China under Grant No. 11271210, K.C.Wong Magna Fund in
Ningbo University. Jipeng Cheng and Kelei Tian are supported by the National Natural Science Foundation of China under Grant No. 11301526, 11201451.}

\newpage{}


\end{document}